\documentclass[10pt,conference]{IEEEtran}
\usepackage{latexsym,amssymb,amsmath,graphicx,bbm}
\usepackage{epic}
\usepackage{psfrag}
\usepackage{amsbsy}
\usepackage{mathptm} 
\usepackage[bbgreekl]{mathbbol}
\usepackage{epsfig}
\usepackage{cite}
\usepackage[small,normal,bf,up]{caption2}
\usepackage{setspace}

\newcommand{\D}[2]{\frac{\partial #1}{\partial #2}}

\newcommand{\eps}{\epsilon}
\newcommand{\brc}[1]{\left({#1}\right)}

\newcommand{\bfx}{{\bf x}}
\newcommand{\bfy}{{\bf y}}

\DeclareMathOperator{\reals}{\ensuremath{\mathbb{R}}}

\DeclareMathOperator{\prob}{{\text{\rm P}}}    

\newcommand{\dr}{{\tt r}}

\newcommand{\dl}{{\tt l}}

\newcommand{\ledge}{\lambda}           
\newcommand{\redge}{\rho}              

\newcommand {\barr}{\begin{array}}
\newcommand {\earr}{\end{array}}

\newcommand{\graphtextsize}{\small}

\newtheorem{lemma}{Lemma}
\newtheorem{theorem}{Theorem}
\newtheorem{example}{Example}

\begin{document}

\title{Existence Proofs of Some EXIT Like Functions}

\author{
\authorblockN{Vishwambhar Rathi}
\authorblockA{School of Computer and Communication Sciences\\
EPFL \\
vishwambhar.rathi@epfl.ch}
\and
\authorblockN{Ruediger Urbanke}
\authorblockA{School of Computer and Communication Sciences \\
EPFL \\
ruediger.urbanke@epfl.ch}
}

\maketitle

\begin{abstract}
The Extended BP (EBP) Generalized EXIT (GEXIT) function introduced in
\cite{MMRU05} plays a fundamental role in the asymptotic analysis of sparse
graph codes. For transmission over the binary erasure channel (BEC) the
analytic properties of the EBP GEXIT function are relatively simple and well
understood. The general case is much harder and even the existence of the curve
is not known in general. We introduce some tools from non-linear analysis which
can be useful to prove the existence of EXIT like curves in some cases. The
main tool is the Krasnoselskii-Rabinowitz (KR) bifurcation theorem.
\end{abstract}

\section{Introduction}
The Extended BP (EBP) GEXIT function introduced in \cite{MMRU05} plays an important
role in the analysis of iterative coding systems. For transmission over the binary
erasure channel (BEC) this function encodes both the behavior of the BP as well
as the MAP decoder in the asymptotic limit of infinite blocklengths. 
Further, in this case the EBP GEXIT function has a very simple analytic expression
in terms of the degree distribution of the ensemble.

It is conjectured that the fundamental characteristic of EBP GEXIT functions 
remains valid also for general (binary memoryless symmetric) channels. 
Figure~\ref{fig:gexiteg} shows the EBP GEXIT
function for the degree distribution pair $(\lambda(x) = 0.25 x+0.75 x^7$, $\rho(x) = x^7)$,
assuming that transmission takes place over the binary symmetric channel (BSC).
\begin{figure}[htp]
\centering
\input{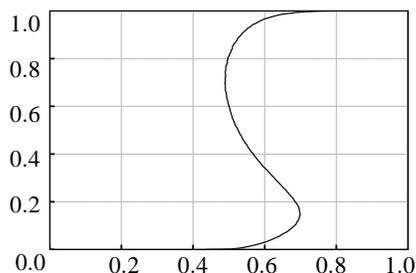}
\vspace*{0.2cm}
\caption{\label{fig:gexiteg} EBP GEXIT function for $\lambda(x) = 0.25 x+0.75 x^7$, $\rho(x) = x^7$
and Binary Symmetric Channel.}
\end{figure}
\vspace*{-0.1in}
Note that this curve smoothly connects the point $(1, 1)$, corresponding to 
the channel BSC$(\frac12)$, with the point $(h^{\text{stab}}, 0)$, where
$h^{\text{stab}}$ corresponds to that channel parameter at which the coding
system changes its stability behavior. The curve was computed using a 
procedure suggested in \cite{MMRU05}. 

This procedure guarantees in general the existence of a fixed point density for
every point on the vertical axis. Unfortunately, it does {\em not} 
guarantee that the set of fixed points so computed forms a smooth one-dimensional
manifold. Such a property however, is required in order to complete  
the theory of EBP GEXIT functions. E.g., it is known that if the curve is smooth
then the area it encloses is equal to the code rate. Combined with the General Area Theorem
(first proved for the BEC in \cite{AKB04} and then extended to the general case
in \cite{MMRU05}) this statement on the area gives rise to bounds on the MAP performance
for sparse graph codes. For the BEC it has been shown that in many cases the bound
is tight and it is conjectured to be tight not only for the BEC but also in the general
case.
 
The existence of the EBP GEXIT function is therefore a fundamental
question at the heart of the asymptotic theory of sparse graph codes.
We introduce some
tools from non-linear analysis which can be useful to prove the existence
of EXIT like curves in some cases. The main tool is the
Krasnoselskii-Rabinowitz (KR) bifurcation theorem.

\section{Definitions and Theorem Related to the Existence of Fixed Points}
As discussed in the last section, it is a difficult task to prove
the existence of the EBP GEXIT curve for general channels. I.e., it is 
difficult to prove that the set of fixed point densities of density evolution
forms a differentiable one-dimensional manifold. 

Although we currently do not know how to prove the existence for the general
case, a fundamental theorem of non-linear analysis, 
called the Krasnoselskii-Rabinowitz (KR) theorem (\cite{Bro03book}, \cite{Kes04book}),
can be helpful in some instances to establish the existence of 
an unbounded connected component of fixed points. To be more precise:
density evolution represents a non-linear map in the space of densities.
If we are given a degree distribution pair with a non-zero fraction
of degree-two variable nodes and a family of BMS channels, then this
map has a {\em bifurcation} point for that channel parameter which
corresponds to the stability condition. In other words, consider the channel parameter  
for which the linearization of the density evolution map around the density corresponding to 
perfect decoding has its largest eigenvalue equal to one. Then this channel 
parameter is a bifurcation point.
Under some technical conditions
the KR theorem then guarantees that there is a connected set of
fixed points which starts at this bifurcation point and which either extends
to infinity or which connects back to another bifurcation point. 
This is not quite as strong a statement as we would wish: 
we are not guaranteed that this connected set forms a smooth manifold, nor do
we know that the curve connects to the fixed point corresponding to the worst density 
and worst channel. Nevertheless, if the theorem applies, we at least know the existence 
of the EBP GEXIT curve locally around the stability point. Before we can show some
cases where the KR theorem can be applied let us quickly review 
the main notation and the main statement.

We denote a generic Banach space by $X$ (e.g. $X=\reals^N$).  
We denote elements of $X$ in boldface letters, i.e., ~$\bfx \in X$. 
We denote the space of bounded linear
operators from $X$ to $X$ by $L(X)$. We are interested in maps of the form
$G:\reals \times X \to X$. The argument $\gamma$ of $G(\gamma, \bfx)$ is called
the {\em parameter}. In our setting the parameter will be the {\em channel parameter}
(e.g., the erasure probability of the BEC or the cross-over probability for the BSC).
Recall the following definitions:
\begin{itemize} 
{\item Completely Continuous (CC) Map:}
A map $G:\reals \times X \to X$ is CC if it maps every bounded set $A$ of
$\reals \times X $ to a relatively compact set in $X$.  
{\item Frechet differentiable:}
Let $G:\reals \times X \to X$ be a map such that $G(\gamma,{\bf 0})={\bf 0}$.
$G$ is Frechet differentiable at $\bfx={\bf 0}$ if there exists $T \in L(X)$
such that, given $\eps > 0$ and an interval $[\gamma_0, \gamma_1]$ of $\reals$,
there exists $\delta > 0$ with the property that $||\bfy|| < \delta$ implies 
\[
\frac{||G(\gamma, \bfy)-\gamma T \bfy||}{||\bfy||} < \epsilon 
\]  
for all $\gamma
\in [\gamma_0,\gamma_1]$. Note that $\delta$ depends on both the choice of
interval and the value of $\eps$. We say that $\gamma T$ is the Frechet
derivative of $G$ at ${\bf 0}$.
\end{itemize} 
We denote the set of non trivial fixed points of $G$ by $S =
\{(\gamma, \bfx): G(\gamma, \bfx) = \bfx, \bfx \neq {\bf 0}\}$ and the closure of $S$ by
$\overline{S}$. If a point $(\mu, {\bf 0}) \in \overline{S}$, then the number $\mu$
is called a {\it bifurcation point} for the solutions to $G(\gamma, \bfx) = \bfx$.
\begin{theorem}[KR Theorem]
\protect{\cite[Theorem~17.8]{Bro03book}}
\label{thm:kr}
Let $X$ be a Banach space and let $G : \reals \times X \to X$ be a map. 
 Let $S = \{(\gamma, \bfx): G(\gamma, \bfx) 
= \bfx, \bfx \neq 0\}$ be the set of non trivial fixed points of $G$ and 
let $\overline{S}$ denote the closure of $S$. 
Assume that the following hypothesis holds.
\begin{enumerate}
\item $G(\gamma, \bfx)$ is a completely continuous map. 
\item $G(\gamma, \bfx)$ is Frechet differentiable at ${\bf 0}$, with 
Frechet derivative $\gamma T$. 
\item Let $\frac{1}{\mu}$ be an eigenvalue of $T$ which is of odd algebraic multiplicity.
\end{enumerate}
Then there exists a maximal closed connected subset $C_{\mu}$ of $\overline{S}$ which 
contains $(\mu, {\bf 0})$ and one of the following is true.
\begin{enumerate}
\item $C_\mu$ is unbounded in $\reals \times X$.
\item $C_\mu$ contains $(\mu^*, {\bf 0})$ for some other bifurcation point $\mu^* \neq \mu$.
\end{enumerate}
\begin{figure}[ht]
\begin{center}
\begin{picture}(130,90)
\put(-20,0){\includegraphics[width=2.0in, height=1.3in]{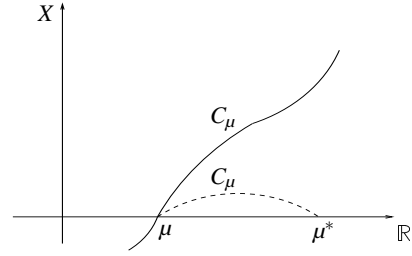}}
\put(128,7){\makebox(0,0){\small{$\reals$}}}
\put(-7,90){\makebox(0,0){\small{$X$}}}
\put(38,7){\makebox(0,0){\small{$\mu$}}}
\put(98,7){\makebox(0,0){\small{$\mu^*$}}}
\put(60,50){\makebox(0,0){\small{$C_\mu$}}}
\put(60,27){\makebox(0,0){\small{$C_\mu$}}}
\end{picture}
\end{center}
\caption{The solid curve shows how the component $C_\mu$ would look like if the first 
conclusion of theorem holds and the dotted one shows the how the component $C_\mu$ would 
look like if the second conclusion holds.}
\label{fig:krpic}
\end{figure}
\end{theorem}
A graphical representation of the KR theorem is shown in Figure~\ref{fig:krpic}. 
 
Our basic plan of attack is the following. In our setting $\bfx$ will denote
a density, and $G$ will be the density evolution map. We want to parametrize the
space in such a way that ${\bf 0}$ denotes the desired fixed point corresponding
to perfect decoding. The parameter $\mu$ will parametrize the channel.
If we can show that the linearization of the density evolution map around ${\bf 0}$
has eigenvalue $1/\mu$, where $\mu$ denotes
the channel parameter which corresponds to the stability condition, and 
if the linearization fulfills the desired technical conditions,
then there is a connected component of fixed-points which either extends to infinity
or is connected to another bifurcation point. At least locally, we will therefore
have proved the existence of a connected component of fixed points. 

In the following it is also good to know the following fact.
\begin{theorem}
\cite[Theorem~17.4]{Bro03book}
\label{thm:bpleqev}
Let $G : \reals \times X \to X$ be a completely continuous and Frechet
differentiable at $\bf{0}$, with derivative $\gamma T$. If $\frac{1}{\mu}$
is not an eigenvalue of the compact linear operator $T$, then there exist
$\epsilon, \eta > 0$ such that $G(\gamma, x) \neq x$ for all $(\gamma, x)$ for
which $|\gamma - \mu| < \epsilon$ and $0 < ||x|| < \eta$. In particular, $\mu$ is
not a bifurcation point for the solutions to $G(\gamma, x) = x$.
\end{theorem}

We also use the following terminology in the rest of the paper.
Let $G:\reals \times \reals^N \to \reals^N$ be a map of the 
form $G=\{G_i\}_{i=1}^N$, where
$G_i:\reals \times \reals^N \to \reals$ is a multivariate polynomial in the
components of $\bfx$ and $G_i(\mu,\bfx) = G_i^1(\bfx) + \mu G_i^2(\bfx)$. Then we say that
$G$ is a vector polynomial map.
\section{Examples}
In principle, we would like to apply the KR theorem directly to the BP or min-sum
decoder. But there are some technical conditions that make the direct
application difficult. For example, the bifurcation point for the
BP decoder appears when the Bhattacharyya parameter is equal to 
$\frac{1}{\lambda'(0) \rho'(1)}$. 
This suggest that the Bhattacharyya parameter should play the role
of the parameter in the setting of the KR theorem. The theorem requires
that the parameter $\gamma$ takes on values in $\reals$ and not only
on $[0, 1]$. Therefore, we can not just work in the space of symmetric
densities (for which the Bhattacharyya parameter is in the range $[0, 1]$)
but we are required to extend the space. How this is best done is currently an open
question.  Because of these technical difficulties, we consider 
quantized decoders. First we show the application of the KR theorem to the simplest possible 
case. 

\begin{example}[BP Decoder for Binary Erasure Channel]
It is instructive (and easy) to analyze the fixed points of the density evolution map for 
the BEC$(\epsilon)$. Consider a degree distribution pair $(\ledge, \redge)$ with
$\lambda'(0) \rho'(1) > 0$.

The density evolution recursion reads
\begin{equation*}
x_{l} = \epsilon \lambda\brc{1-\rho\brc{1-x_{l-1}}}.  
\end{equation*} 
We take the space $X$ to be $X=\reals$ and set $G(\epsilon, x) = \epsilon
\lambda\brc{1-\rho\brc{1-x}}$.  Here the erasure probability $\epsilon$ plays
the role of the parameter. As $G(\epsilon, x)$ is a polynomial map, it is
completely continuous by Lemma~\ref{lem:cc} and Frechet differentiable by
Lemma~\ref{lem:fd}. From Lemma~\ref{lem:fd}, the Frechet derivative of
$G(\epsilon, x)$ is given by $\epsilon T x=\epsilon \lambda'(0) \rho'(1) x$.
Thus the parameter $\epsilon$ appears multiplicatively, as required by the KR
theorem.

Trivially, $\lambda'(0) \rho'(1)$ is the eigenvalue of the operator $T$ and
this eigenvalue has multiplicity one (the space is only one-dimensional), which
is odd.  Since by assumption $\lambda'(0) \rho'(1) > 0$, this eigenvalue is
strictly positive. Thus $1/\brc{\lambda'(0) \rho'(1)}$ is a bifurcation point.
As there can be only one eigenvalue of $T$, there can be at most one
bifurcation point (Theorem~\ref{thm:bpleqev}).  Thus the first conclusion of
Theorem~\ref{thm:kr} holds true: the connected component of fixed points
containing the bifurcation point $\brc{\frac{1}{\lambda'(0) \rho'(1)}, \bf{0}}$
is unbounded.  \end{example}

Of course, for this simple example we even have an explicit characterization of
this connected set of fixed points and an application of the powerful KR
theorem is not needed. But for only slightly more elaborate examples an explicit
characterization is typically no longer available. 

Consider now transmission over the Binary Symmetric Channel (BSC) with
transition probability $p$ and min-sum (MS) decoding. For iteration $l$, let
$M^{(l)}_{m->n}$ be the message sent from check node $m$ to variable node $n$
and $M^{(l)}_{n->m}$ be the message sent from variable node $n$ to check node
$m$. We denote the set of neighbors of a node $m$ by $\mathcal{N}(m)$. If we
assume that we represent messages as log-likelihood ratios then the processing
rules in each iterations are as follows: 
\begin{enumerate}
\item Processing rule at check nodes---for each $m$ and each $n \in \mathcal{N}(m)$, 
\begin{equation}\label{eqn:checkrule}
M^{(l)}_{m->n} = \prod_{n' \in \mathcal{N}(m)\\n} \text{sgn}\brc{M^{(l)}_{n'->m}} 
		\text{min}_{n' \in \mathcal{N}(m)\\n} \left|M^{(l)}_{n'->m}\right|
\end{equation}
\item Processing rule at variable nodes---for each $n$ and each $m \in \mathcal{N}(n)$,
\begin{equation}\label{eqn:varrule}
M^{(l)}_{n->m}=L_n + \sum_{m' \in \mathcal{N}(n)\\m} M^{(l-1)}_{m'->n}, 
\end{equation}
where $L_n$ denotes the initial log-likelihood ratio received by node $n$.
\end{enumerate}

We claim that there exist a one-to-one mapping between the messages of the min-sum decoder and 
the set of integers ${\mathbb Z}$. More precisely, the messages of the min-sum decoder are
of the form $i \ln \frac{1-p}{p}, i \in {\mathbb Z}$. This can be easily seen
by induction.  The initial messages from the variable nodes to the check nodes
are $\pm \ln \frac{1-p}{p}$. At the check nodes if all the incoming messages
are of the form $i \ln \frac{1-p}{p}$, then by inspecting the check node
processing rule given in Equation~(\ref{eqn:checkrule}) we see that the
outgoing message is again of this form. At the variable nodes, all the messages
are added up which clearly preserve this property. 
We can therefore
equivalently formulate message-passing under min-sum on the lattice
$\mathbb{Z}$ by assuming that the initial messages are from the set $\{ \pm
1\}$ and have probabilities $(1-p)$ and $p$, respectively.

In order to be able to apply the KR theorem, below we consider {\em
bounded} versions of min-sum, i.e., we bound the absolute value of the messages
to $M$, where $M$ is a fixed integer.  More precisely, we assume that message
alphabet is $\mathcal{M}=\{-M, -(M-1), \cdots, -1, 0, 1, \cdots, M-1, M\}$.  As
mentioned before, $\mathcal{M}_c = \{-1, 1\}$.  The message passing rule for
the check node side is the same as given by Equation~(\ref{eqn:checkrule}). On
the other hand, to enforce the boundedness constraint, we need to slightly
modify the message-passing rule for variable nodes. For a node of degree  $d_v$
the rule is defined by: 
\begin{equation}\label{eqn:vnodemap}
\Psi_v\brc{m_0, m_1, \cdots, m_{d_v-1}} = \left\{ 
\begin{array}{ll}
	  & \hspace*{-0.4cm}\textrm{$\exists i$ s.t.~$m_i=M$} \\ 
    M	  & \hspace*{-0.4cm}\textrm{$\nexists j$ s.t.~$m_j \neq -M$}\\ 
	  & \\
	  & \hspace*{-0.4cm}\textrm{$\exists i$ s.t.~$m_i=-M$}\\
   -M 	  & \hspace*{-0.4cm}\textrm{$\nexists j$ s.t.~$m_j \neq M$}\\ 
	  & \\
	  & \hspace*{-0.4cm}\textrm{$\exists i, j$ s.t.}\\
    0	  & \hspace*{-0.4cm}\textrm{$m_i=-M, m_j=M$}\\
	  & \\
	\mathcal{Q}\brc{\sum_{i=0}^{d_v-1} m_i} & \textrm{otherwise},
\end{array}\right.
\end{equation}
where the quantization function $\mathcal{Q}(x)=M$ if $x \geq M$,  
$\mathcal{Q}(x)=-M$ if $x \leq -M$ and equal to $x$ otherwise.
Note that the exact rule for the case when both $M$ and $-M$ are incoming to
the variable node is not really important since this should hardly ever happen 
if $M$ is large enough. This is because if $M$ is large, the quantized decoder 
will mimic more and more the min-sum decoder. 

For future reference, consider the ensemble $\brc{\Lambda(x) = 0.4 x^2 + 0.6
x^5, \Gamma(x) = x^4}$.  It has design rate $r=0.05$. The Shannon threshold for
this rate is $p^{\text{Sh}} = 0.369$.  Table~\ref{tab:msquantthresh} shows the
threshold values of this ensemble for increasing values of $M$ as well as the
threshold under true min-sum decoding. We see that the thresholds for finite
$M$ quickly converge to the unbounded case.
\vspace*{-0.1in}
\begin{table}[ht]
\centering
\begin{tabular} {ccccccc}
$M$ & 1 & 2 & 3 & 4 & 5 & $\infty$\\  
$\approx p^*$ & $0.0319$ & $0.0962$ & $0.0974$ & $0.1219$ &  \hspace*{-0.3cm} $0.1318$ & \hspace*{-0.25cm}$0.148$\\
\end{tabular}
\caption{\label{tab:msquantthresh} 
Thresholds of $\Lambda(x) = 0.4 x^2 + 0.6 x^5$,
$\Gamma(x) = x^4$ under quantized min-sum decoding.} 
\end{table}
\vspace*{-0.1in}
Note that this quantizer and the message passing rules satisfy the symmetry
conditions of \cite{RiU01}. Thus we can perform the density evolution under the
all-one codeword assumption. Recall that the alphabet has $2M+1$ elements.  But
since the probability of the individual elements sums up to one, the density
evolution recursion $G$ can be written as a function of $2 M$ variables.  Thus,
the underlying space is $X=\reals^{2 M}$. As can be easily seen, the density
evolution map is again a vector polynomial map.  Thus such a map is completely
continuous by Lemma~\ref{lem:cc} and Hypothesis $1$ of Theorem~\ref{thm:kr} is
satisfied. The first condition for the second hypothesis to hold true is that
$G\brc{p,\bf{0}}=\bf{0}$. Note that $\bf{x}=0$ implies that with probability
one, the message is equal to $M$. Now at the check node side if all the
incoming messages are equal to $M$, then the outgoing is also equal to $M$. The
same holds true for the variable node side by the definition of $\Psi_v$
given in Equation~(\ref{eqn:vnodemap}).  Also the channel transition
probability $p$ appears only as $p$ and $1-p$. Thus the Frechet derivative of
the map is of the form $p T + T'$, where both $T, T' \in \reals^{2 M \times 2
M}$. In order to satisfy Hypothesis $2$ of Theorem~\ref{thm:kr}, we need to
modify the density evolution map.  We use Lemma~\ref{lem:affine} and
consider the derived map with Frechet derivative $p (I_{2 M}-T')^{-1}T$. 
\begin{example}[Min-Sum Decoder with $M=2$]\label{eg:m2}
For our running example consider $M=2$.
The Frechet derivative is of the form $p T
+ T'$, where $T'$ is not identically zero.  Fortunately $(I_4-T')^{-1}$ exists.
As mentioned before, by Lemma~\ref{lem:affine} we need to study the
eigenvalues of the matrix $(I_4-T')^{-1} T$. The matrix $(I_4-T')^{-1} T$ has
eigenvalues $\frac{1}{\mu_1}=3.50027$, $\frac{1}{\mu_2}=-2.70249$ and the other
two eigenvalues are zero.  Both $\frac{1}{\mu_1}$ and $\frac{1}{\mu_2}$ have
multiplicity one (i.e., the multiplicities are odd).  This implies that the KR
theorem is applicable to both the eigenvalues and at least one of the
conclusion of the KR theorem must hold true for both of them. In particular
$(\mu_1, {\bf 0})$ and $(\mu_2, {\bf 0})$ are bifurcation points. Let
$C_{\mu_1}$ and $C_{\mu_2}$ be the fixed point component containing $\mu_1$ and
$\mu_2$ respectively.  Now by the KR theorem either the fixed point connected
component $C_{\mu_1}$ and  $C_{\mu_2}$ are unbounded or $C_{\mu_1}= C_{\mu_2}$.

We can compute the fixed points explicitly in this case. The result is shown in
Figure~\ref{fig:qpt5}. Since the fixed points are elements of $\reals^4$ we
need to project them into $\reals$ in order to be able to plot them. We choose
to apply the error probability operator. As the density evolution is done
assuming that the all-one codeword has been transmitted, so the error
probability operator sums up the component corresponding to negative indices
and adds to this sum half the weight of index zero as it is like an erasure.
\begin{equation}\label{eqn:peop}
\prob_e\brc{\bf{x}} = \sum_{i=-M}^{-1} x_i + \frac{x_0}{2}.
\end{equation}
As we can see, the second conclusion of Theorem~\ref{thm:kr} holds
i.e.~$C_{\mu_1}=C_{\mu_2}=C_\mu$.  The fixed point connected component $C_\mu$
containing the point a$=\brc{\mu_1, \bf{0}}=\brc{0.28569, \bf{0}}$ also
contains the point d$=\brc{\mu_2, \bf{0}}=\brc{-0.37003, \bf{0}}$. In the
component $C_\mu$, the branch from a to b is stable, b to c is unstable and c
to d is stable. The component $C'$ is stable. The threshold is $p^*=0.0962$.
The fixed point of iterative decoder at the threshold is represented by point e
of the fixed point component $C'$. Above the threshold, the fixed points of
iterative decoder moves upward along $C'$ as the channel transition probability
$p$ increases.  

\begin{figure}[htp]
\centering
\input{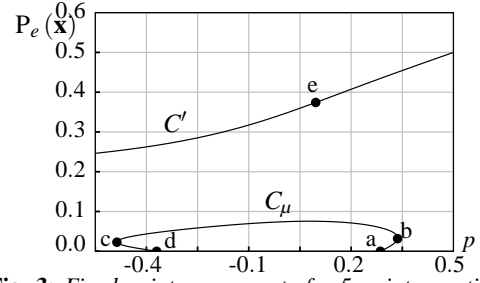}
\caption{\label{fig:qpt5} Fixed point components for $5$ point quantizer.} 
\end{figure}
\end{example}

\begin{example}[Min-sum decoder with $M=3$]
For our running example we consider~$M=3$. The Frechet derivative is again of
the form $p T + T'$. In this case also the inverse $(I_6-T')^{-1}$ exists. By
Lemma~\ref{lem:affine}, we need to study the eigenvalues of  $(I_6-T')^{-1} T$.
The matrix  $(I_6-T')^{-1} T$ has the only non-zero real eigenvalue as
$\frac{1}{\mu}=2.09804$ and its multiplicity is one. So the  KR theorem is
applicable in this case. Note that as there is only one non-zero eigenvalue,
there can be at most one bifurcation point by Theorem~\ref{thm:bpleqev}. Thus
the second conclusion of KR theorem can not be true. This implies that the
first conclusion holds: there is an unbounded component $C_{\mu}$ of fixed
point containing the bifurcation point $\mu$. In this case also we can compute
this component explicitly. As the fixed points are element of $\reals^6$, in
order to plot them we project them to one dimension by the error probability
operator given in Equation~\ref{eqn:peop}. The plot is shown in
Figure~\ref{fig:qpt7}. The bifurcation point is
a=$\brc{\mu, 0.0}=(0.476636, 0.0)$.  As far as the stability of the fixed
point in $C_\mu$ is concerned, the branch a to b is stable.  The fixed points
in branch b to c is unstable and from point c onwards the fixed points are
stable.  The point e represents the fixed point at which the iterative decoder
get stuck at threshold $p^* \approx 0.0974$. Above the threshold, the fixed
points of iterative decoder moves upward along $C_\mu$ as the channel
transition probability $p$ increases.
\vspace*{-0.1in}
\begin{figure}[htp]
\centering
\input{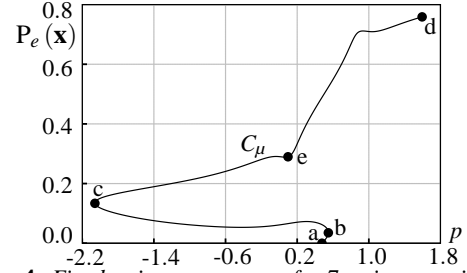}
\caption{\label{fig:qpt7} Fixed point components for $7$ point quantizer.} 
\end{figure}
\end{example}
\vspace*{-0.1in}
Discussion: We presented the examples $M=2$ and $M=3$. It is tempting
to increase $M$ and see how the fixed point structure changes. 
By taking $M$ to infinity, one would hope to recover  the structure of the 
fixed point components of the un-quantized min-sum decoder.

\begin{example}[Decoder with Erasure]
The decoder with erasure was introduced in \cite{RiU01}. 
The underlying channel is $BSC(p)$. On the variable node side the message-passing rule for
a node of degree $\dl$ reads
\begin{equation}
\Psi_v\brc{m_0,m_1,\cdots,m_{\dl-1}} = \mathrm{sgn}\brc{m_0+\sum_{i=1}^{\dl-1}
m_i}.\nonumber
\end{equation}
The rule for a check node of degree $\dr$ is
\begin{equation}
\Psi_c(m_1,\cdots,m_{\dr-1}) = \prod_{i=1}^{\dr-1} m_i.\nonumber
\end{equation}
Note that for this decoder if there are degree two variable nodes then the
threshold is $0$ i.e.~${\bf 0}$ can not be a fixed point. To see this, suppose
that all the incoming messages to variables nodes are equal to one.  Then with
probability $p$, the outgoing message from a variable node is equal to $0$.
Thus the probability of $0$ is equal to $\lambda_2 p$. Hence we assume that
$\lambda_2=0$. For this example $M=1$, hence the underlying space is
$X=\reals^2$. The density evolution equation can be found in \cite{RiU01}.  The
Frechet derivative of the density evolution map can again be computed and it
turns out that its only eigenvalue is $2 \lambda_3 \rho'(1)$.  But now this
eigenvalue has even multiplicity.  So we can not apply the KR theorem to this
case. In \cite{MaW04}, it was investigated if the conclusions of the KR theorem
is still applicable to an eigenvalue of even multiplicity. We are currently
investigating whether the result of \cite{MaW04} is applicable to the decoder
with erasure.  However numerical computation of fixed point suggest that indeed
$\frac{1}{2 \lambda_3 \rho'(1)}$ is a bifurcation point. For example, in
Figure~\ref{fig:dwe} we plot the fixed point component of $(3, 6)$ regular
ensemble. For this ensemble $2 \lambda_3 \rho'(1)=10$, so supposedly $p=0.1$ is
a bifurcation point. We can see from Figure~\ref{fig:dwe} that point a which
corresponds to $p=0.1$ is indeed a bifurcation point. The threshold for this
ensemble is $p^*=0.0708$.\footnote{This assumes that in the first iteration we set the
weight of the channel to $2$ and in all subsequent iterations to $1$.} The point b represents the fixed point at which
decoder get stuck at the threshold. The branch a to b is unstable. From b
onwards the fixed points are stable.
\vspace*{-0.15in}
\begin{figure}[htp]
\centering
\input{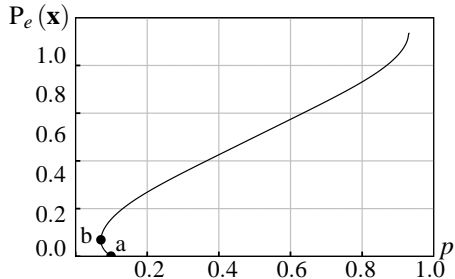}
\caption{\label{fig:dwe} Fixed point component for the decoder with erasure 
for $(3, 6)$ LDPC ensemble.} 
\end{figure}
\end{example}
\vspace*{-0.25in}
\section{Outlook}
We have shown how the tools of non-linear analysis can be used in proving the
existence of fixed points. Our ultimate goal is to understand the fixed point structure of 
the BP and the min-sum decoder. For the min-sum decoder we hope to accomplish our goal
by considering a sequence of quantized decoders where the number of quantization points
tends to infinity. Whether a similar strategy can be devised for the BP decoder is
still an open question.
\vspace*{-0.05in}
\appendix
\begin{lemma}\label{lem:cc}
Every vector polynomial map $G: \reals \times \reals^N \to \reals^N$ is a
completely continuous map.  
\end{lemma}
\begin{proof}
Consider any bounded set $S$ in $\reals^N$. As $S$ is bounded, 
so will be all the components $G_i(S)$. Hence the
 set $G(S)$ is also bounded. Clearly this would imply that the closure
$\overline{G(S)}$ is also bounded. In a finite dimensional vector space a
closed and bounded set is a compact. Hence $G(S)$ is relatively compact. Thus
the map $G$ is Completely continuous. 
\end{proof}
\begin{lemma}\label{lem:fd}
Let $G:\reals^N \to \reals^N$ be a vector polynomial map such that $G({\bf 0}) = {\bf 0}$.
Then $G$ is Frechet differentiable. The Frechet derivative $T$ of $G$ is a
matrix whose entries are given by $\{t_{ij}\}$ where $1 \leq i,j \leq N$ and 
\[
 \left. t_{ij} = \D{G_i}{x_j}\right\rvert_{\bfx={\bf 0}}.
\] 
\end{lemma}
\begin{proof}
Consider $||G(\bfx)-T \bfx||$. As $|\bfx_i| \leq ||\bfx||$, there are no linear term in
$G(\bfx)-T \bfx$ and $G({\bf 0})={\bf 0}$ implies that $||G(\bfx)-T \bfx|| = o\brc{||\bfx||^2}$. Hence
\[
 \frac{||G(\bfx)-T \bfx||}{||\bfx||} = o\brc{||\bfx||}.
\]
This proves the lemma.
\end{proof}

Note that Hypothesis 2 of Theorem~\ref{thm:kr} implies that the parameter 
$\gamma$ must appear multiplicatively in the Frechet derivative. But in many cases 
we see that the Frechet derivative is of the form $\gamma T + T'$. The following 
lemma says that in this case also the KR theorem can be applied provided the linear 
operator $I-T'$ is invertible.
\begin{lemma}\label{lem:affine} 
Let $G: \reals \times \reals^n \to \reals^n$ be a vector polynomial map and
Frechet differentiable with Frechet derivative $\gamma T + T'$. Let us assume that
$\brc{I_n-T'}^{-1}$ exists.  Let $F(\gamma, {\bf x}) \triangleq \brc{I_n-T'}^{-1}
\brc{G(\gamma, {\bf x})-T'{\bf x}}$. Then $F$ is a vector polynomial map and Frechet 
differentiable with Frechet derivative $\gamma \brc{I_n-T'}^{-1} T$. Also the set of 
fixed points of $F$ is same as set of fixed points of $G$. 
\end{lemma}
\begin{proof}
The fact that $F$ is a vector polynomial map is obvious. For the Frechet
differentiability of $F$ we need that $F(\gamma, {\bf 0}) = {\bf 0}$.  Now,
$F(\gamma, {\bf 0})=\brc{I_n-T'}^{-1} \brc{G(\gamma, {\bf 0})-T'{\bf 0}} = {\bf 0}$,
as $G(\gamma, {\bf 0})={\bf 0}$. Now the Frechet derivative of $\brc{G(\gamma ,{\bf
x})-T'{\bf x}}$ is given by $\gamma T \bfx$. This implies that the Frechet
derivative of $F\brc{\gamma, \bfx}$ is equal to $\gamma \brc{I_n-T'}^{-1} T$. To see
that $F$ and $G$ have the same set of fixed points, let $\bfx$ be a fixed point 
of $G$. Then $G\brc{\gamma, \bfx} - T' \bfx = \bfx - T' \bfx$ which implies  
$\brc{I_n-T'}^{-1} \brc{G\brc{\gamma, \bfx} - T' \bfx} = \bfx$ 
i.e.~$F\brc{\gamma, \bfx}  = \bfx$.
\end{proof}
\section*{Acknowledgment}
Many thanks to the member of non linear analysis reading group: Nicolas Macris, 
Shrinivas Kudekar, Satish Babu Korada, Sanket Dusad, Dinkar Vasudevan, Harm Cronie. 
Many thanks to Andrea Montanari for helpful discussions.

\end{document}